\documentclass[UKenglish,cleveref,autoref,a4paper]{article}
\usepackage[utf8]{inputenc}
\usepackage{babel}
\usepackage{lmodern}
\usepackage[T1]{fontenc}
\usepackage{graphicx}
\usepackage{microtype}
\usepackage{parskip}
\usepackage{setspace}
\usepackage[dvipsnames]{xcolor}
\usepackage{hyperref}
\usepackage{subcaption}
\usepackage[left=3cm,right=3cm,top=3cm,bottom=3cm]{geometry}

\usepackage{dsserif}  % pour mettre le 1 en gras de tableau (cf https://www.ctan.org/topic/font-bbd)
\usepackage{amsmath}
\usepackage{amssymb}
\usepackage{amsfonts}
\usepackage{amsthm}
\newtheorem{theorem}{Theorem}
\newtheorem{proposition}[theorem]{Proposition}
\newtheorem{corollary}[theorem]{Corollary}
\newtheorem{lemma}[theorem]{Lemma}
\newtheorem{definition}[theorem]{Definition}
\newtheorem{algorithm}[theorem]{Algorithm}
\newtheorem{remark}[theorem]{Remark}

\theoremstyle{plain}

\graphicspath{{figures/}} 

\newcommand{\HH}{\mathbb{H}^2}  % plan hyperbolique
\newcommand{\arcsinh}{\operatorname{arcsinh}}  
\newcommand{\arccosh}{\operatorname{arccosh}}
\newcommand{\surf}{S}  % surface
\newcommand{\e}{\varepsilon}  % epsilon
\newcommand{\thin}{\surf^t_\e}  % partie eps-fine
\newcommand{\thick}{\surf^T_\e}  % partie eps-épaisse
\newcommand{\dist}{\delta_\surf}  % distance sur la surface
\newcommand{\disth}{\delta_{\HH}}  % distance dans H^2
\newcommand{\h}{\widetilde}  % tilde pour les relevés dans H
\newcommand{\D}{\mathcal{D}_{\widetilde{b}}}  % domaine de Dirichlet
\newcommand{\sys}{\sigma}  % systole
\newcommand{\g}{\gamma}  % isométrie
\newcommand{\id}{\mathbb{1}_\Gamma}  % identité de Gamma
\newcommand{\diam}{\operatorname{Diam}}  % diamètre
\newcommand{\tr}{\operatorname{Tr}}  % trace

\renewcommand{\leq}{\leqslant}
\renewcommand{\geq}{\geqslant}

\title{Computing an $\e$-net of a closed hyperbolic surface}
\author{Vincent Delecroix, Vincent Despré, Camille Lanuel,\\ Hugo Parlier, Monique Teillaud}
\date{}

\begin{document}

\maketitle

\begin{abstract}
Hyperbolic surfaces are a fundamental object in mathematics and play an increasingly important role in computational geometry and topology. A key ingredient in the design of efficient algorithms on such surfaces is the availability of a geometric discretization of controlled complexity. In this paper, we present the first algorithm for constructing $\e$-nets on hyperbolic surfaces starting from a fundamental polygon representation. Our approach is based on Delaunay refinement and relies on maintaining Delaunay triangulations through edge flips.

The size of an $\e$-net cannot be bounded solely as a function of the genus because of the presence of arbitrarily long collars around short geodesics. To overcome this difficulty, we introduce the notion of a pseudo $\e$-net, which decomposes the surface into $\e$-thin cylinders together with a Delaunay triangulation over an $\e$-net of the remaining thick part.

As applications, we obtain algorithms for computing the length spectrum of an $\e$-thick hyperbolic surface and for computing the systole from a pseudo $\log(\sqrt{2})$-net. These results demonstrate that Delaunay-based discretizations provide a practical and versatile framework for algorithmic computations on hyperbolic surfaces.

\end{abstract}
 
\section{Introduction}
This paper focuses on hyperbolic surfaces, that is, surfaces endowed with a Riemannian metric of constant negative curvature $-1$. These surfaces have been extensively studied because of their ubiquitous nature. Beyond the realm of geometric topology in which they naturally live, they are important objects in spectral geometry, geometric group theory and the study of their closed geodesics is often shows intriguing relations to number theory and dynamics.  

Hyperbolic geometry also plays a central role in theoretical computer science. A classical example arises in the study of rotation distance between binary trees~\cite{sleatorRotationDistanceTriangulations1986}. More recently, hyperbolic geometry has become a powerful tool for graph representation and visualization~\cite{eppsteinSquarepantsTreeSum2009,eppsteinLimitationsRealisticHyperbolic2021}. Furthermore, the hyperbolic plane provides the canonical model of the universal cover of surfaces of genus at least $2$. This perspective has proved instrumental in the design and analysis of algorithms on surfaces, as well as in the proof of purely topological results~\cite{despreComputingGeometricIntersection2019,deverdiereUntanglingGraphsSurfaces2024}. Recently, the hyperbolic geometry of surfaces has been a surprisingly strong tool to attack problems on graph drawings \cite{HubardParlier2025,HubarddeMesmayParlier2025}. This follows a tradition of using hyperbolic geometry in combinatorics, namely since the work of Sleater, Tarjan and Thurston \cite{sleatorRotationDistanceTriangulations1986,sleatorRotationDistanceTriangulations1988}.

To perform algorithmic computations on a hyperbolic surface $\surf$, one first needs a suitable computational representation. From this perspective, a fundamental polygon is a natural choice, as it closely relates to the combinatorial maps (typically triangulations) commonly used to represent surfaces in computational geometry and topology. A fundamental polygon in $\HH$ simultaneously encodes both the topology of the surface through the side identifications and its hyperbolic geometry through the embedding of the polygon.

A first major step in the algorithmic treatment of hyperbolic surfaces was the development of an algorithm computing a Dirichlet domain from an arbitrary fundamental polygon~\cite{despreComputingDirichletDomain2023}. Since Dirichlet domains enjoy strong geometric properties, they provide a convenient canonical representation from which subsequent algorithms can be developed.

The next challenge is to obtain a local geometric description of the surface. A natural discretization is provided by an $\e$-net, namely a set of points that is simultaneously $\e$-dense and $\e$-separated. Such nets are ubiquitous in computational geometry, as they capture the local geometry while remaining of controlled size. The first contribution of this paper is an algorithm for constructing $\e$-nets on hyperbolic surfaces.

Buser proposed a closely related construction in the form of triangulations with controlled geometry~\cite[Section 4.5]{buserGeometrySpectraCompact2010}. Although his construction is not algorithmic, it suggests a natural strategy: start from a sparse set of points and iteratively insert the centers of large empty disks until the desired density is achieved. While this observation may appear elementary, it strongly motivates the use of Delaunay triangulations, which identify maximal empty disks as a by-product of their construction.

This is precisely the principle underlying Shewchuk's Delaunay refinement algorithm in the Euclidean setting~\cite{shewchukDelaunayRefinementAlgorithms2002}. Delaunay triangulations have since been extensively studied in hyperbolic geometry, both in the hyperbolic plane and on hyperbolic surfaces~\cite{bogdanovHyperbolicDelaunayComplexes2014,iordanovImplementingDelaunayTriangulations2017,despreFlippingGeometricTriangulations2020,ebbensMinimalDelaunayTriangulations2022}. Building on these developments, we design a hyperbolic analogue of Shewchuk's refinement algorithm, using the flip algorithm to maintain the Delaunay triangulation throughout the refinement process.

\begin{theorem}
    Let $\surf$ be a hyperbolic surface of genus $g$ and sytole $\sys$ given by a Dirichlet domain. We can compute an $\e$-net on $S$ in $O_{\e\to\infty}(f(g,\sys)/\e^4)$ time where:
    \begin{itemize}
        \item $f(g,\sys)=O_{g\to\infty}(e^{40g^2})$
        \item $f(g,\sys)=O_{\sys\to\infty}((\log(1/\sys)/\sys)^{6g})$
    \end{itemize}
\end{theorem}

The complexity of the refinement algorithm depends on the size of the resulting $\e$-net. Unfortunately, for a fixed value of $\e$, this size cannot be bounded independently of the surface. Indeed, as the systole tends to zero, the Collar Lemma implies the existence of increasingly long embedded cylinders around short geodesics. Consequently, an arbitrarily large number of net points may be required to cover these cylinders. We recall the Collar Lemma precisely in the background section.

To overcome this difficulty, Buser introduced the notion of a trigon decomposition, which partitions the surface into thin cylinders and a thick part that can be triangulated efficiently~\cite[Section 4.5]{buserGeometrySpectraCompact2010}. Inspired by this construction, we introduce the notion of a pseudo $\e$-net. Rather than covering the entire surface with an $\e$-net, a pseudo $\e$-net decomposes $\surf$ into a collection of $\e$-thin cylinders together with a Delaunay triangulation over an $\e$-net of the remaining $\e$-thick part.

\begin{theorem}
\label{thm:complexite_bananes}
	Let $\surf$ be a hyperbolic surface given by a Delaunay triangulation on a single vertex (or, equivalently, by a Dirichlet domain) and let $\e\leq\ln\sqrt 2$ be a positive number. A pseudo $\e$-net of $\surf$ can be computed in $O(f(g,\sys)/\e^4)$ time.
\end{theorem}

Beyond providing a convenient geometric representation of hyperbolic surfaces, $\e$-nets also form the basis of efficient algorithms for computing geometric invariants. Our first application is the computation of the length spectrum, that is the collection of lengths of the closed geodesics of the surface. As hinted at previously, the length spectrum is arguably the most important geometric invariant of a hyperbolic surface, as it encodes the dynamics of the geodesic flow, but is also equivalent, by a celebrated result of Huber \cite{huberAnalytischenTheorieHyperbolischer1959}, using the Selberg trace fomula \cite{selbergHarmonicAnalysisDiscontinuous1956}, to the Laplace spectrum. 

We show that, on an $\e$-thick surface, a Delaunay triangulation over an $\e$-net provides enough geometric and combinatorial information to enumerate all closed geodesics up to a prescribed length.

\begin{theorem}\label{thm:spectrum}
Let $\e>0$, let $\surf$ be an $\e$-thick hyperbolic surface, and let $T$ be a Delaunay triangulation over an $\e$-net of $\surf$. The length spectrum of $\surf$ up to length $L$ can be computed in $O(g\,m\,L\,e^L)$ time, where $m$ denotes the maximum multiplicity of a length in the spectrum below $L$.
\end{theorem}

We use the maximum multiplicity in the above estimate, but it can be replaced by a function of $L$ and the genus (see Remark \ref{rem:multiplicity}). The pseudo $\e$-net decomposition allows us to extend this approach to surfaces with arbitrarily small systole (see Theorem \ref{thm:spectrum2}). In particular, choosing $\e=\log(\sqrt{2})$ yields a polynomial-time algorithm for computing the systole once a pseudo $\log(\sqrt{2})$-net has been constructed.

\begin{corollary}
    Let $\surf$ be an $\log(\sqrt{2})$-thick surface and $T$ a Delaunay triangulation over a pseudo $\log(\sqrt{2})$-net of $\surf$. We can compute the length of the systole of $\surf$ in $O(g^2)$ time.
\end{corollary}

The remainder of the paper is organized as follows. In Section~2, we review the necessary background on hyperbolic surfaces and establish the notation used throughout the paper. Section~3 develops the structural properties of hyperbolic surfaces and of their $\e$-nets that underlie our algorithms. In Section~4, we present our algorithm for constructing $\e$-nets based on Delaunay refinement, prove its correctness and complexity. Section~5 introduces pseudo $\e$-nets and explains how they allow us to handle the $\e$-thin parts of hyperbolic surfaces. Finally, Section~6 applies these constructions to the computation of the systole and the length spectrum of a hyperbolic surface.

\noindent {\bf Acknowledgements.}
Hugo Parlier was supported by ANR-SNF Grant number 200021E\_238147
(SUGAR). Vincent Delecroix and Vincent Despré were partly supported by grant ANR-25-CE40-0416 of the French National Research Agency (SUGAR).
Vincent Despré was also partly supported by grant ANR-23-CE48-0017 of the French National Research Agency ANR (project Abysm).

\section{Background on hyperbolic surfaces and notation}\label{sec:background}
We refer the reader to textbooks for more details, e.g. \cite{buserGeometrySpectraCompact2010, beardonGeometryDiscreteGroups1983}.

A closed hyperbolic surface can be seen as the quotient $\HH/\Gamma$ of the hyperbolic plane $\HH$ under the action of a group $\Gamma$ of orientation-preserving isometries.
Throughout the paper, objects in $\HH$ are denoted with a tilde $\h \cdot$, while objects on $\surf$ are denoted without. In particular, for an object $o$ on $\surf$, $\h o$ denotes any of its lifts in $\HH$. To simplify the language, we often use the term \emph{copy} to refer to an image of an object in $\HH$ by an element of $\Gamma$. 

We work with the Poincaré disk model in which the hyperbolic plane $\HH$ is represented as the unit disk of the complex plane $\mathbb{C}$. The unit circle consists of points at infinity.  The geodesics are either diameters of the unit disk, or circular arcs that meet the boundary circle orthogonally. The hyperbolic circles are Euclidean circles (but their hyperbolic and Euclidean centers differ). Orientation-preserving isometries are represented as matrices in $\mathbb{C}^{2\times 2}$. 

\subsection{Delaunay triangulation and Dirichlet domain}\label{sec:rep_surface}

A triangulation $T$ of $\surf$ is a partition of $\surf$ into triangles; note that edges may be loops. A triangulation of $\surf$ is a Delaunay triangulation if for each triangle $t$ of $T$ and any of its lifts $\h t$ in $\HH$, the open disk circumscribing $\h t$ contains no vertex of the (infinite) lift of $T$ in $\HH$~\cite{despreFlippingGeometricTriangulations2020}. 
The Voronoi diagram is the dual of the Delaunay triangulation. The Dirichlet domain $\mathcal{D}_{\h{x}}$ of a point $\h{x}\in\HH$ is the (closed) cell of $\h{x}$ in the Voronoi diagram of its (infinite) orbit $\Gamma\h{x}$. Unlike the Euclidean case, $\Gamma$ is non-commutative, and the combinatorics of a Dirichlet domain depends on the point $x$ (Figure~\ref{fig:Dirichlet}). The number $k$ of sides of $\mathcal{D}_{\h{x}}$ satisfies $4g\leq k \leq 12g-6$ (see, e.g.,~\cite{despreComputingDirichletDomain2023}).

\begin{figure}[!ht]
	\centering
	\begin{subfigure}{0.5\textwidth}
		\centering
		\includegraphics[width=.95\textwidth]{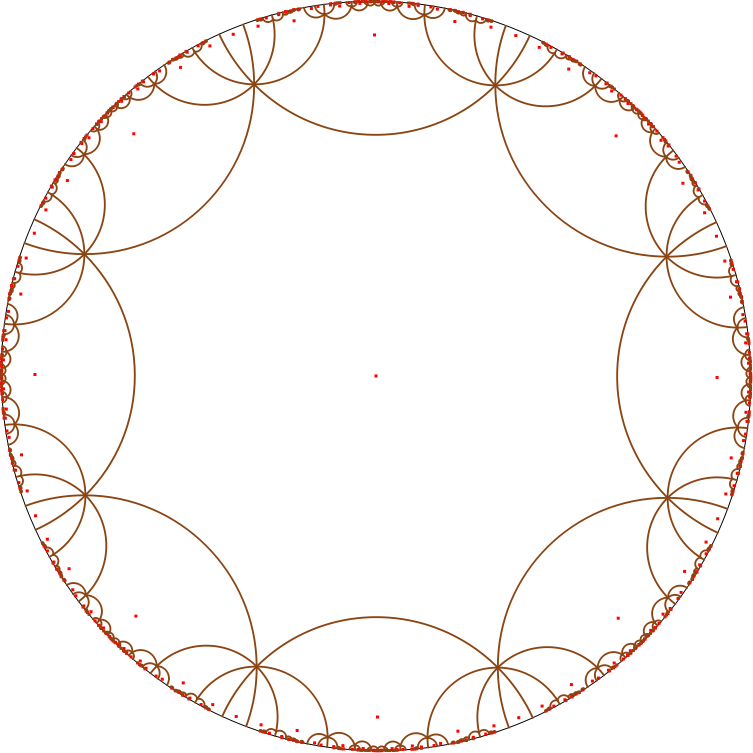}
	\end{subfigure}%
	\begin{subfigure}{0.5\textwidth}
		\centering
		\includegraphics[width=.95\textwidth]{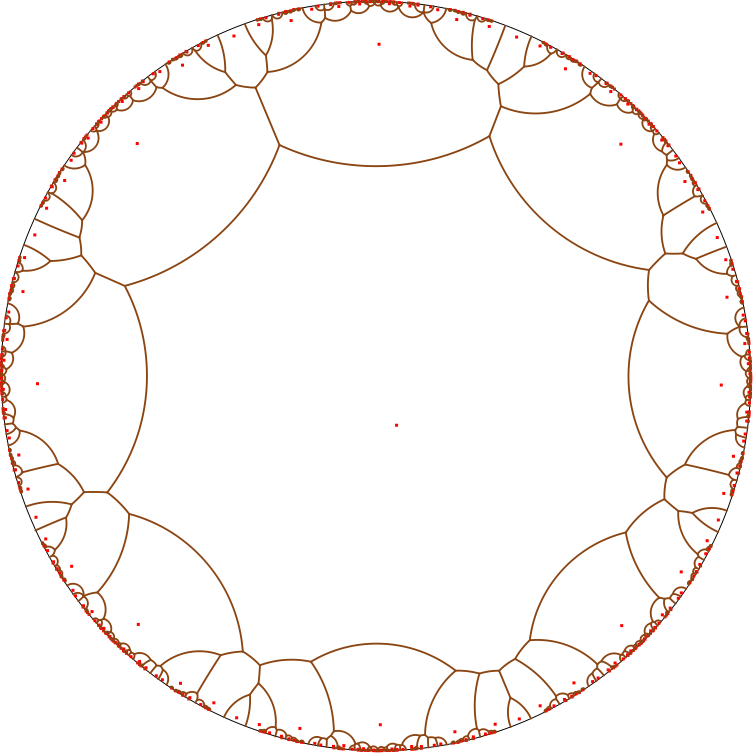}
	\end{subfigure}
				
	\caption{Dirichlet domains for the Bolza surface ($g=2$). The domain on the left has $4g=8$ sides and the one on the right has $12g-6=18$ sides. Figure from~\cite{bogdanovDelaunayTriangulationsOrientable2016}.}\label{fig:Dirichlet}
\end{figure}

In this paper, we assume that the input surface $\surf$ is given by a Delaunay triangulation having a single vertex $b$, i.e., all Delaunay edges are loops based in $b$. The point $b$ is arbitrary. This introduces no restriction, as such a representation can be computed for any closed hyperbolic surface, starting from a standard representation by a fundamental domain and side pairings~\cite{despreComputingDirichletDomain2023}.\footnote{The common basepoint is denoted as $b''$ in~\cite{despreComputingDirichletDomain2023}.} The Dirichlet domain $\D$ of some lift $\h{b}$ of $b$ can be computed together with the corresponding side pairings, which are generating the group $\Gamma$. The sides of $\D$ are denoted as $s_i, i=0,\ldots,k-1$ and the corresponding side pairings as $\g_i, i=0,\ldots,k-1$ (here, side pairings are pairwise inverses).

\subsection{Thick and thin parts}\label{sec;thinthick}

The \emph{injectivity radius} $r_x(\surf)$ of $\surf$ at a point $x$ is the supremum of all $r>0$ such that the open ball of radius $r$ centered at $x$, $B(x, r)=\left\{ y\in \surf \mid \dist(x, y) <r\right\}$, where $\dist$ is the distance on $\surf$, is isometric to a disk in $\HH$. In particular, $B(x,r)$ is a topologically embedded disk on $\surf$ for all $r\leq r_x(\surf)$. The \emph{systole} $\sys$ of a surface is the length of its shortest non-contractible curve. By a slight abuse of notation, we denote both the curve itself and its length by $\sys$. For closed hyperbolic surfaces, the systole is directly related to the injectivity radius via $\sys=2\cdot\inf \left\{ r_x(\surf)\mid x\in \surf\right\}$.

Similarly, for a simple closed geodesic $\gamma$, it is interesting to know if it can be embedded in a metric cylinder of a certain width. The classical collar lemma~\cite{buserGeometrySpectraCompact2010} states that $\gamma$ is necessarily included in an embedded cylinder, often called its collar, of width $2\cdot\arcsinh\left(1 / \sinh \left(\frac{1}{2} \ell\left(\gamma\right)\right)\right)$. We will call curves that are shorter than their collar width \emph{small} curves. Thus a curve $\gamma$ is deemed to be be small if and only if it satisfies $\ell(\gamma)\le2\cdot\arcsinh\left(1 / \sinh \left(\frac{1}{2} \ell\left(\gamma\right)\right)\right)$ which in turn amounts to $\ell(\gamma)\le2\cdot\arcsinh(1)$. 

A crucial property of small curves, that motivates the definition, is that they cannot pairwise intersect. In this context, we consider the collars of small curves to be \emph{thin} regions, while the rest of the surface is considered \emph{thick}. It allows an adaptation of the collar lemma, specifically for small curves:
\begin{theorem}[The collar lemma for small curves~{\cite[thm 4.1.6]{buserGeometrySpectraCompact2010}}]\label{thm:collar}
	Let $\surf$ be a closed hyperbolic surface of genus $g \geq 2$, and let $\gamma_1, \ldots, \gamma_m$ be small simple closed geodesics on $\surf$. Then the following hold:
	\begin{enumerate}
		\item[(i)] $m \leq 3 g-3$ and the curves are disjoint.
		
		\item[(ii)] $r_p(\surf)>\arcsinh(1)$ for all $p$ in the the thick part of the surface.
		
		\item[(iii)] If $p$ is in the collar of $\gamma_i$, at distance $d$ of $\gamma_i$ then 
		\[
		\sinh(r_p(\surf))=\sinh \left(\frac{1}{2}\ell(\gamma_i)\right)\cdot\cosh(d)
		\]
	\end{enumerate}
\end{theorem}
It is important to notice that the injectivity radius of a point $p$ in the collar of $\gamma_i$ is half of the length of the geodesic loop based at $p$ and freely homotopic to $\gamma_i$. We have also slightly modified formula (iii), using the distance from the core geodesic of the collar and not its boundary as it appears in the referenced book.

In the context of $\e$-nets, it is meaningful to consider thin parts with respect to a parameter $\e$. For any $\e>0$, the \emph{$\e$-thin part} of $\surf$ is $\thin = \left\{ x \in \surf \mid r_x(\surf) \leq\e/2\right\}$, and its \emph{$\e$-thick part} is $\thick=\surf\setminus \thin$. Observe that if $\e<\sys$, then there is no $\e$-thin part.

\begin{figure}[!ht]
	\begin{center}
		\includegraphics{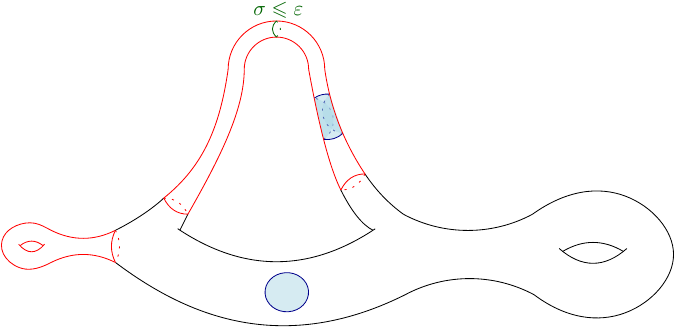}
		\caption{Thick and thin (red) parts of a hyperbolic surface. Disks of radius $\e$ are shown in blue.}\label{fig:thickthin}
	\end{center}
\end{figure}

\section{Structural properties of hyperbolic surfaces and their $\e$-nets}\label{sec:bound}

\subsection{Number of points in an $\e$-net}
    
Let $P$ be an $\e$-packing of $\surf$. The open balls of radius $\e /2$ centered at the points of $P$ on the $\e$-thick part $\thick$ are isometric to disks in $\HH$ and are pairwise disjoint. The area of such a disk centered at a point $p$ is $\mathcal{A}\left(B\left(p, \frac \e 2 \right)\right) = 4\pi\sinh^2\left(\frac \e 4\right)$~\cite[Theorem 7.2.2]{beardonGeometryDiscreteGroups1983}. Since $\sinh(x)\geq x$ for all $x\geq 0$, we have $\mathcal{A}\left(B\left(p, \frac \e 2 \right)\right) \geq \pi \e^2/4$.
	
Let $N^T$ be the number of points of $P$ while lie on the $\e$-thick part $\thick$. By the Gauss-Bonnet theorem, the area of the surface $\surf$ is $\mathcal{A}(\surf) = 4\pi(g-1)$. Summing the above inequality over all the points in $P\cap \thick$, we obtain
$N^T \pi\e^2/4 \leq \sum_{p\in P\cap \thick}\mathcal{A}\left(B\left(p, \frac\e 2 \right)\right) \leq 4\pi(g-1)$, and thus we have the following bound.
\begin{lemma}\label{eq:nb_points_epaisse} 
	$N^T\leq \frac{16(g-1)}{\e^2}$
\end{lemma} 
	        
The open balls of radius $\e/2$ in the $\e$-thin part $\thin$, if they exist (that is if $\sys\leq \e$) are also pairwise disjoint, but they are no longer topological disks (so in particular not isometric to disks in $\HH$). However, by definition, the open balls of radius $\sys/2$ are isometric to disks in $\HH$. When $\sys\leq \e$, we can apply the reasoning that led to inequality~(\ref{eq:nb_points_epaisse}) for $\sys$ instead of $\e$, and obtain a bound on the number of points of $P$ on the thin part $\thin$: $N^t\leq 16(g-1)/\sys^2$. The bound on the total number of points of $P$ follows.

\begin{proposition}\label{prop:nb_points}
	Any $\e$-packing of a closed hyperbolic surface $\surf$ of genus $g$ and systole $\sys$ contains $N \leq 16(g-1)\left(1/\e^2 + 1/\sys^2\right)$ points. If $\e < \sys$, then $N \leq 16(g-1)/\e^2$.
\end{proposition}

\subsection{Diameter}
We concentrate on lower bounds for $N$ in terms of our choice of $\e$ and the geometry of $S$. Observe that $\diam(S)$, the diameter of the surface, can be bounded in terms of the number of points and $\e$. As balls of radius $2\e$ around the $N$ points cover the surface, we have
$
\diam(S) < 2\e N$.
It follows that 
$
N> \frac{\diam(S)}{2\e}.$

We now look at the geometry of surface in the neighborhood of a systole. Notice that the shortest path between two points on a systole is itself an arc of the systole, otherwise one could construct a shorter non-contractible curve. The same argument can be extended to a collar region around the systole: the half-collar (that is the collar of width $\frac{1}{2}\arcsinh\left(1/(\sinh((\sys/2))\right)$) is also geodesically convex. In particular, this means that
$$
\diam(S)> \arcsinh\left(1/(\sinh((\sys/2))\right).
$$
This latter quantity, for very small systole, behaves roughly like $\log\left(\frac{1}{\sys}\right)$. Conversely, we can bound $\diam(\surf)$ by above w.r.t. $\sys$.

\begin{lemma}\label{lem:diam_sys}
Let $\surf$ be a closed hyperbolic surface of genus $g$, and let
$\sys=\sys(\surf)$. Then
\[
    \diam(\surf)
    \le
    (3g-3)\left(13+\log_+(1/\sys)\right),
\]
where $\log_+(x)=\max\{0,\log x\}$.
\end{lemma}

\begin{proof}
Let $p$ be a path realizing $\diam(\surf)$. Since $p$ is a shortest path,
the cylinder of width $\sys/4$ around $p$ is embedded in $\surf$. As the
area of this cylinder is bounded above by the area of $\surf$, we obtain
\[
    \diam(\surf)\cdot \frac{\sys}{4}\le 4\pi(g-1).
\]
Hence, if $\sys\ge 2\arcsinh(1)$, then
\[
    \diam(\surf)
    \le \frac{8\pi(g-1)}{\arcsinh(1)}
    < 13(3g-3),
\]
and the desired estimate follows since $\log_+(1/\sys)\ge0$.

We may therefore assume that $\sys<2\arcsinh(1)$.

We now assume that $\sys < 2\arcsinh(1)$. In this case, we decompose $p$ into two collections of subpaths: a subpath $p_\sys$ contained in the collars of the short curves of $\surf$, and a subpath $p'$ contained in the thick part. Arguing as above, we obtain the estimate
$\ell(p') \le \frac{8\pi(g-1)}{\arcsinh(1)}$.

Let $\gamma$ be a boundary component of a collar. Suppose that $p$ intersects $\gamma$ at least three times, and denote by $v_0$, $v_1$, and $v_2$ the first three intersection points along $p$ for some fixed orientation. Then the subpaths from $v_0$ to $v_1$ and from $v_1$ to $v_2$ must themselves be shortest paths. However, the shortest path between two points on $\gamma$ is entirely contained in the corresponding collar~\cite[Lemma 4.1.5]{buserGeometrySpectraCompact2010}, which yields a contradiction.

It follows that $p$ intersects each boundary component of a collar at most twice. In particular, each collar contributes to $p_\sys$ either through subpaths homotopic to portions of its boundary or through paths crossing the collar. Indeed, $p$ cannot cross the same short curve twice without being shortened by replacing the corresponding segment with a subarc of that curve. Finally, we note that the situation where one endpoint of $p$ lies inside a collar leads to a simpler configuration and can be treated in the same way.

When a component of $p_\sys$ intersects a short curve, the intersection is not necessarily orthogonal. The configuration that maximizes the length occurs when the path connects opposite points on the boundary components of the collar. In this case, the path is homotopic to the concatenation of a segment realizing the width of the collar and a geodesic segment connecting opposite boundary points.

The length of the geodesic loop homotopic to a short curve $s$, based at a point on the boundary of its collar, is $2\arcsinh(\cosh(\ell(s)/2))$ by Theorem~\ref{thm:collar}. This quantity is increasing in $\ell(s)$ and is therefore maximized when $\ell(s)=2 \arcsinh(1)$. Similarly, the length of a geodesic segment between opposite boundary points of a collar is bounded above by $2\arcsinh(\sqrt{\sqrt{2}-1})$. Consequently, we obtain
$\ell(p_\sys)\le (3g-3)\left(2\arcsinh(\sqrt{\sqrt{2}-1}) + \arcsinh(1/\sinh(\sys/2))\right)$,
since there are at most $3g-3$ collars on $\surf$ and $\arcsinh(1/\sinh(\sys/2)) > 2\arcsinh(\sqrt{\sqrt{2}-1})$.

Combining the estimates for $p'$ and $p_\sys$, we obtain
$$\ell(p)\le (3g-3)\left(2\arcsinh(\sqrt{\sqrt{2}-1}) + \arcsinh(1/\sinh(\sys/2))\right) + \frac{8\pi(g-1)}{\arcsinh(1)}.$$

A direct estimate of the preceding expression gives
\[
    \ell(p)\le (3g-3)\left(13+\log(1/\sys)\right)
\]
whenever $\sys\le1$. If $1<\sys<2\arcsinh(1)$, the same explicit expression is uniformly bounded above by $13(3g-3)$. Thus, in both cases,
\[
    \diam(\surf)=\ell(p)
    \le (3g-3)\left(13+\log_+(1/\sys)\right).
\]
\end{proof}
\begin{remark}
    When we do not need an explicit bound and when $\sys$ goes to 0, we have $\diam(\surf)=O(g*\log(1/\sys))$ and at fixed genus $\diam(\surf)=\Theta(\log(1/\sys))$.
\end{remark}

\subsection{Counting curves}

A standard approach to control the length of simple closed curves is to describe them via their Dehn–Thurston coordinates and to use the fact that the length of a curve is bounded below by a linear function of these coordinates~\cite{flp}. A similar strategy applies to paths instead of closed curves, as explained in~\cite{despreFlippingGeometricTriangulations2020}. Here, we refine this approach by making the dependence on $g$ and $\sys$ explicit.

Let $\mathcal{B}$ be a Bers pants decomposition of $\surf$, that is, a pants decomposition minimizing the maximum length of its curves among all such decompositions. Furthermore, we suppose that every curve in the pants decomposition is locally minimizing: that means that if we remove one curve, there is no strictly shorter curve it can be replaced by to complete the set of curves into a pants decomposition. By \cite{parlierShorterNoteShorter2024}, such a pants decomposition exists with all lengths $\le 4\pi(g-1)$.

Consider the Dehn–Thurston coordinates $(m_i,t_i,s_i)_{i\in[1,\dots,3g-3]}$ of a path $p$.

Fix a curve $\gamma$ in $\mathcal{B}$. It bounds two pairs of pants, whose union forms a sphere with four boundary components or it bounds one pair of pants, when it is an interior curve to a one-holed torus. In both cases, there are three combinatorial ways to partition these boundary components into two pairs. In each case, there exists a shortest curve realizing the partition (one of them being $\gamma$, by the local minimality of $\mathcal{B}$). We denote by $(m,t,s)$ the intersection numbers of $p$ with these three curves, where $m$ corresponds to $\gamma$, and the coordinates satisfy one of the relations $m = t + s$, $t = m + s$, or $s = m + t$.

Up to the choice of endpoints of $p$ on the boundary components, a path is determined by its coordinates $(m_i,t_i,s_i)_{i\in[1,\dots,3g-3]}$. Hence, for a fixed set of coordinates, there are at most $\sum m_i$ such paths. Moreover, for given $m_i$ and $t_i$, there are at most three possible values of $s_i$.

Let $b$ denote the length of the longest curve in $\mathcal{B}$. As discussed previously, we have $b \leq 4\pi(g-1)$. Any non-trivial simple arc in a pair of pants, with endpoints on its boundary, has length at least the width of the collars of the curves it traverses. This yields the lower bound $w := 2\arcsinh(\sinh^{-1}(b/2)) > 2e^{-b/2}$, which implies the following lemma.

\begin{lemma}
A simple path on $\surf$ of length $l$ intersects $\mathcal{B}$ at most $l/w + 1$ times, that is, $\sum m_i \le l/w + 1$.
\end{lemma}

\begin{proposition}\label{prop:nb_curves}
Let $\surf$ be a surface of genus $g$, let $u,v \in \surf$, and let $\ell>0$. Then there are at most $(\ell/w+1)^{3g-1}(6\ell/\sys+1)^{3g-3}$ distinct simple geodesic paths of length at most $\ell$ between $u$ and $v$ on $\surf$.
\end{proposition}

\begin{proof}
We count paths by fixing their intersection pattern with $\mathcal{B}$, that is, by fixing the values of $m_i$ and the combinatorics of how the path connects to $\mathcal{B}$.

Consider first the twist parameters $t_i$. In each pair of adjacent pairs of pants, there is a minimal contribution to $t_i$: each time a path connects two boundary components lying on different sides of the curve corresponding to $t_i$, it contributes at least one unit. Increasing $t_i$ by one increases the length of the path by at least $\sys/2$. It follows that, for fixed $m_i$, the number of possible values of each $t_i$ is bounded by $2\ell/\sys+1$.

We now estimate the total number of possibilities. First, the number of choices for the $m_i$ is bounded by $(\ell/w+1)^{3g-3}$, since each $m_i$ is non-negative and $\sum m_i \le \ell/w+1$. For the $t_i$, we obtain at most $(2\ell/\sys+1)^{3g-3}$ possibilities. For each pair of adjacent pairs of pants, there are at most three possible values of $s_i$, yielding a factor $3^{3g-3}$.

It remains to account for the way the endpoints $u$ and $v$ are connected to $\mathcal{B}$. This contributes at most $(\ell/w+1)^2$ additional choices.

Putting everything together, we obtain at most
$3^{3g-3}(\ell/w+1)^{3g-1}(2\ell/\sys+1)^{3g-3}$
homotopy classes of simple paths of length at most $\ell$. This completes the proof.
\end{proof}

\section{Construction of the $\e$-net}\label{sec:firstalgo}
\subsection{Data-structure and algorithm}
The input of the algorithm consists of the Delaunay triangulation of $\surf$ with a single vertex $b$, together with the Dirichlet domain $\D$ of a lift $\h b$ and the group $\Gamma$ generated by side-pairings. As mentioned previously, we are readily able to move between certain ways of representing our surface, so this does not induce a loss of generality.

Our algorithm is inspired by Shewchuk's Delaunay refinement~\cite{shewchukDelaunayRefinementAlgorithms2002}. The general idea is to break each Delaunay triangle whose circumcircle has a radius greater than $\e$ by inserting its circumscribing center in the triangulation. 

We reuse the data structure proposed by Despré \emph{et al.} for computing the Delaunay triangulation of a surface by edge flips~\cite{despreFlippingGeometricTriangulations2020}. A triangulation of $\surf$ is represented by 
\begin{itemize}
	\item its vertices: a vertex $p$ has constant-time access to its lift $\h{p_b}$ in $\D$ and one of its incident triangles;
	\item and its triangles: a triangle $\Delta$ has constant-time access to its three vertices $p_0^\Delta, p_1^\Delta, p_2^\Delta$, its three adjacent triangles, and three isometries $\g_0^\Delta=\id, \g_1^\Delta, \g_2^\Delta$ in $\Gamma$ defined as follows. 
\end{itemize}
A triangle $\Delta=(p_0^\Delta; p_1^\Delta; p_2^\Delta)$ does not always have a lift entirely included in $\D$. However, it always has at least one lift with at least one vertex in $\D$ (see Figure~\ref{fig:data_structure}). Let us choose such a lift and denote it as $\h{\Delta_0}$; up to a re-indexing of its vertices, $\h{p_0^\Delta}\in\D$. Then $\g_1^\Delta$ and $\g_2^\Delta$ are the isometries such that the other two vertices of $\h{\Delta_0}$ are $\g_1^\Delta\h{p_1^\Delta}$ and $\g_2^\Delta\h{p_2^\Delta}$. Note that the other lifts of $\Delta$ having at least one vertex in $\D$ can be retrieved by applying the inverses of these isometries to $\h{\Delta_0}$.
The union, on all triangles of the triangulation of $\surf$, of their lifts with at least one vertex in $\D$ covers the fundamental domain $\D$.
 
\begin{figure}[!ht]
	\centering
	\includegraphics{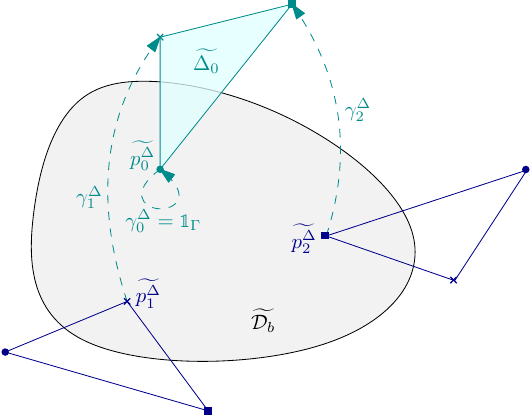}
	\caption{Example of a triangle $\Delta$ having three lifts with one vertex in $\D$ (the hyperbolic triangles are schematically represented with straight edges).}\label{fig:data_structure}
\end{figure}

We denote by $DT(\cdot)$ a Delaunay triangulation of a set of points on $\surf$. 

Let us fix  $\e > 0$. In a first step, the set of points is initialized as $P_1=\{b\}$.

At each step $i\geq 2$, the algorithm inserts the circumscribing center $c$ of a triangle $\Delta^\e$ whose radius is greater than $\e$. The set of points is updated as $P_i=P_{i-1}\cup \{c\}$ as well as the Delaunay triangulation $DT(P_i)$. To do so, several operations are needed. 

We first compute the radius of $\h{\Delta_0}$ for every triangle $\Delta$ of $DT(P_{i-1})$, until a triangle $\Delta^\e$ whose radius is greater than $\e$ is found.\footnote{Of course a priority queue could be used to improve the complexity of this search. We accept a linear complexity for simplicity, as this is not the dominant operation in the algorithm.} The circumcenter $\h c$ of the lift $\h {\Delta^\e_0}$ is a lift of $c$, but it does not necessarily lie in $\D$. This can be checked by testing whether $\h {b}$ and $\h c$ lie on the same side of the supporting line of each side of $\D$.

To actually insert $c$ into $DT(P_{i-1})$, we need to find the lift $\h c_b$ of $c$ that lies in $\D$. 
If $\h c$ lies in $\D$, then $\h {c_b}= \h c$. Otherwise, the algorithm walks in the tiling $\{\g\D\}_{\g\in\Gamma}$ of $\HH$ along the geodesic segment $\h{p_0^{\Delta^\e}}\h{c}$.
The first copy of $\D$ traversed by $\h{p_0^{\Delta^\e}}\h{c}$ is found by looking for the side $s_{j_1}, j_1\in\{0,\ldots,k-1\}$ of $\D$ intersecting it.\footnote{To check whether two geodesic segments $\h x_1 \h x_2$ and $\h y_1 \h y_2$ intersect, we check whether $\h x_1$ and $\h x_2$ lie on opposite sides of the supporting line of $\h y_1 \h y_2$, and we run the same test, swapping the roles of $x$ and $y$.}
The walk along $\h{p_0^{\Delta^\e}}\h{c}$ continues in $\g_{j_1}\D$, and so on, until the copy $\g_{j_n}\ldots \g_{j_1}\D$ containing $\h c$ is found. Then $\h c_b = \g_{j_1}^{-1}\ldots \g_{j_n}^{-1}\h c$. Note that the walk still works when $\h{p_0^{\Delta^\e}}\h{c}$ goes through a vertex of a copy of $\D$.

The Delaunay triangulation $DT(P_i)$ of $P_i=P_{i-1}\cup\{c\}$ can then be computed. First, the triangle $\Delta_c$ of $DT(P_{i-1})$ containing $c$ is found by naively checking if $\h{c_b}$ lies in one of the (at most three) lifts of each triangle $\Delta$ in $DT(P_{i-1})$ having a vertex in $\D$. This can be done by testing, for each edge, whether $\h c_b$ and the third vertex of the triangle lie on the same side of its supporting line. Then $\Delta_c$ is split into three by creating an edge between $c$ and its three vertices. In the data structure, the three isometries stored in each new triangle are $\id$ for $c$, and the corresponding isometries in $\Delta_c$ for the other two vertices. Then $DT(P_i)$ is computed with a sequence of flips and the data structure is updated~\cite{despreFlippingGeometricTriangulations2020}.

The termination of the algorithm is quite obvious. At step $i=1$, the $\e$-packing $P_1$ consists of one point. At each step $i\geq 2$, the point added to $P_i$ is the circumcenter of a Delaunay triangle whose radius is at least $\e$. Because no vertex lies in the interior of a Delaunay disk, the center added is at distance at least $\e$ from any point of $P_i$. By induction, $P_i$ is an $\e$-packing containing $i$ points. By Proposition~\ref{prop:nb_points}, the algorithm must terminate after a finite number $N-1$ of insertions. It returns an $\e$-packing $P_N$ of cardinality $N$.
	
It remains to show that $P_N$ is an $\e$-covering of $\surf$. Let $x$ be a point on $\surf$. It lies in a triangle $\Delta$ of $DT(P_N)$. Let $\h \Delta$ be a lift of $\Delta$ and $\h x$ the lift of $x$ lying in $\h \Delta$. The circumdisk of $\h \Delta$ has a radius $r\leq\e$. There is a vertex of $\h{\Delta}$ whose distance to $\h{x}$ is at most $r$ (see Lemma~\ref{lem:point_triangle} in appendix). That vertex is a lift of a point of $P_N$ by definition of $\Delta$. It follows that $\dist(x, P_N)\leq \e$, therefore $P_N$ is an $\e$-net.
This establishes the correctness of our algorithm.

\subsection{Complexity of the location part}\label{sec:complexity}

We first explicit that the following operations take $O(1)$ time in the real RAM model and we consider them as elementary operations:
\begin{itemize}
	\item Computing $\h{\Delta_0}$  from a triangle $\Delta$ of the data structure;
	\item Computing the radius or the center of the circumcircle of a triangle in $\HH$;
	\item Deciding if a point lies on the right or the left side of an oriented geodesic segment in $\HH$;
	\item Flipping an edge of a triangle~\cite[Section  4.1]{despreFlippingGeometricTriangulations2020}. 
\end{itemize}

\begin{lemma}\label{lem:complexityLocate}
    The total cost of the computation and location of circumcenters is $O(N+g^2)$.
\end{lemma}
\begin{proof}
At the beginning of a step $i\geq 2$, $P_{i-1}$ contains $i-1$ points, the Euler characteristic shows that $DT(P_{i-1})$ has $2i+4g-6$ triangles, which gives the cost of finding $\Delta^\e$.
	
Recall that the number $k$ of sides of $\D$ is at most $12g-6$ (see Section~\ref{sec:background}). Determining whether $\h c$ lies in (a given copy of) $\D$ thus requires at most $12g-6$ elementary operations.
The algorithm tests the copies of $\D$ that intersect the geodesic segment $\h{p_0^{\Delta^\e}} \h c$. Since $\h{\Delta^\e_0}$ is a triangle of $DT\left(\h{P_{i-1}}\right)$, its circumcircle does not contain any other lift of $p_0^{\Delta^\e}$, so $\h{p_0^{\Delta^\e}}$ is the closest lift of $p_0^{\Delta^\e}$ to $\h c$. The geodesic segment $\h{p_0^{\Delta^\e}} \h c$ is thus a lift of a distance path\footnote{A distance path on $\surf$ is a shortest path between two points. It is necessarily a geodesic segment, but not all geodesic segments are distance paths since they only locally minimize distances.} on $\surf$, what is called a \emph{distance path} in $\HH$. By~\cite[Proposition 14]{despreRepresentingInfiniteHyperbolic2021}, every side of $\D$ is either a distance path, or the concatenation of two distance paths. As two distance paths that do not have a subarc in common, which is the case here, can intersect at most once~\cite[Lemma 8]{despreRepresentingInfiniteHyperbolic2021}, $\h{p_0^{\Delta^\e}} \h c$ traverses at most $2k$ sides of copies of $\D$. If an intersection occurs at a vertex of degree $d$ of a copy of $\D$, then this counts for $d$ intersections. Searching the copy of $\D$ containing $\h c$ hence requires $k^2\leq (12g-6)^2$ elementary operations. Computing $\h{c_b}$ costs 1 operation. 
		
Finding $\Delta_c$ in $DT(P_{i-1})$ when $\h{c_b}$ is known requires at most $9(2i+4g-6)$ elementary operations since it amounts to checking the three edges of at most three lifts of each triangle. The update of the data structure when splitting the triangle containing $c$ into three is done in $8$ elementary operations (deleting the triangle that contains $c$, adding $c$ to the list of vertices, creating 3 triangles and 3 isometries). 
		
Adding the above costs for step $i$, locating $c$ in $DT(P_{i-1})$ and splitting the triangle containing it costs at most $10(2i+4g-6)+(12g-6)^2+9$ elementary operations. 
\end{proof}

\subsection{Complexity of the flip part} 
The total complexity of the algorithm comes from the sum of those of Lemma~\ref{lem:complexityLocate} and Lemma~\ref{lem:nbflips}, the second one dominating the first one.
	
\begin{lemma}\label{lem:nbflips}
	The total number of flipped edges during the execution of the algorithm is at most $f(g,\sys) 1/\e^4$, where $$f(g,\sys)=256(g-1)^2(3(g-1)\log(e^{13}/\sys)e^{2\pi(g-1)}+1)^{3g-1}(18\pi(g-1)\log(e^{13}/\sys)/\sys+1)^{3g-3}/\sys^4$$ if $\sys\ge2*\arcsinh(1)$ and $$f(g,\sys)=256(g-1)^2(8\pi(g-1)e^{2\pi(g-1)}/\arcsinh(1)+1)^{3g-1}(24\pi(g-1)/\arcsinh(1)^2+1)^{3g-3}$$ if $\surf$ is thick.
\end{lemma}

The proof of this lemma mimicks the proofs in~\cite{despreFlippingGeometricTriangulations2020}.
The situation is quite different here, as the points are inserted incrementally and the flips are done at each insertion, whereas all points are know in advance in~\cite{despreFlippingGeometricTriangulations2020}, which requires to rewrite a complete proof.

\begin{proof}
	Denote as $T_1=DT(P_1), T_2, \ldots, T_K$ the sequence of triangulations appearing during our algorithm. For $j\geq 1$, $T_{j+1}$ is obtained from $T_j$ either by flipping an edge, or by splitting a triangle into three from a new vertex. Every triangulation is \emph{geometric}: it is equivalent to a triangulation whose edges are geodesic segments that do not intersect in their interior. Flipping an edge maintains the property~\cite{despreFlippingGeometricTriangulations2020}; splitting a triangle clearly maintains it, too.
	  
	We associate to any triangulation $T$ of $\surf$ a polyhedral surface $\Sigma$ in $\mathbb{R}^3$ as in~\cite[Section 2.3]{despreFlippingGeometricTriangulations2020}: the vertices of $\Sigma$ are obtained from the vertices of the (infinite) lift $\h T$ of $T$ by the stereographic projection onto the unit sphere $\mathbb{S}^2$. 
	The (infinite) surface $\Sigma$ is convex if and only if $T$ is a Delaunay triangulation of $\surf$. 
	Let us show that $\Sigma_{j+1}$ \emph{contains} $\Sigma_j$ for each $j\geq 1$, i.e., $\Sigma_{j+1}$ lies between $\Sigma_j$ and $\mathbb{S}^2$. 
	The case when $T_{j+1}$ is obtained by flipping a non-Delaunay edge $e$ of $T_j$ is studied in~\cite[Section 2.3]{despreFlippingGeometricTriangulations2020}: $\Sigma_j$ is concave at each edge projected from a lift of $e$, and after the flip $\Sigma_{j+1}$ is convex at all the new edges. Let us now examine the case when $T_{j+1}$ is obtained by splitting a triangle of $T_j$ into three from a new vertex. In this case, $T_j$ is a Delaunay triangulation, so $\Sigma_j$ is convex, but the edges left in $\Sigma_{j+1}$ from each triangle of $T_j$ that is split from the new vertex on $\mathbb{S}^2$ are generally not convex. The surface $\Sigma_j$ is thus contained in $\Sigma_{j+1}$ as well. As a result, every $\Sigma_{j'}$ with $j'>j$ contains $\Sigma_j$. If an edge is flipped at a step $i$ of the algorithm, the corresponding line segment becomes interior to the polyhedral surface and all the following, and it can never reappear.
			
	We can now observe that no flip will ever create an edge longer than $8\diam(\surf)$.
	Consider the fundamental domain $\Omega_{\widetilde{b}}$ consisting of one lift of each triangle of $T_1$ incident to $\h b$.
	For all $\h x\in\Omega_{\widetilde{b}}$, $\disth(\h x,\h b)<2\diam(\surf)$, where $\disth$ is the distance in $\HH$. This is because $\h x$ belongs to the circumdisk of the triangle $\h t \in\h{T_1}$ it lies in, and $\h b$ lies on its boundary. That circumdisk must have a radius $r<\diam(\surf)$, otherwise it would contain at least one lift of every point of $\surf$. In particular, it would contain a lift of a vertex of $\h t$ in its interior, which is impossible since $\h t$ is a Delaunay triangle.
	Let $e$ be an edge created by a flip and let $\h v$ be the lift of its midpoint $v$ lying in $\Omega_{\widetilde{b}}$. The fundamental domain $\Omega_{\widetilde{b}}$ is strictly included in the disk of radius $4\diam(\surf)$ and centered at $\h v$: indeed, if $\h x\in\Omega_{\widetilde{b}}$ then $\disth(\h x, \h v)\leq \disth(\h x, \h b)+\disth(\h b, \h v) < 4\diam(\surf)$. The proof of~\cite[Lemma 10]{despreFlippingGeometricTriangulations2020}, replacing $2\Delta(T)$ by $8\diam(\surf)$ and $\Omega$ with $\Omega_{\widetilde{b}}$, shows that $e$ cannot be longer than $8\diam(\surf)$.

The proof of~\cite[Theorem 19]{despreFlippingGeometricTriangulations2020}, replacing $2\Delta(T)$ with $8\diam(\surf)$, proves that a Delaunay flip algorithm performed on a triangulation of $\surf$ with $N$ vertices flips at most $C_\surf N^2$ edges, where $C_\surf$ is the number of paths smaller than $8\diam(\surf)$. Putting the bound from lemma~\ref{lem:diam_sys} and the upper bound from $N$ from Proposition~\ref{prop:nb_points} into the formula of Proposition~\ref{prop:nb_curves} gives the claimed $f(g,\sys)$.
\end{proof}

\section{Handling $\e$-thin parts}\label{sec:thinparts}
We have shown that the number of points in an $\e$-net of a hyperbolic surface $\surf$ depends on the systole $\sys$ of the surface. Moreover, it becomes arbitrarily large as $\sys$ decreases to zero. This is due to the width of a collar of a small curve, which grows when the curve becomes shorter. It is thus desirable to avoid adding points on these collars. In this section, we explain how to restrict the computation of an $\e$-net to the $\e$-thick part $\thick$ of $\surf$, resulting in what we call a pseudo $\e$-net. The number of points in a pseudo $\e$-net then only depends on the genus of the surface and not on its systole. Since collars are isometric to cylinders, using a pseudo $\e$-net instead of an $\e$-net for approximation algorithms would not result in a significant loss of geometric information.

\begin{definition}[Pseudo $\e$-net]
	Let $\surf$ be a hyperbolic surface and $\e>0$. A \emph{pseudo $\e$-net} of $\surf$ is a finite set of points $P \subset \thick$ such that $P$ is an $\e$-net of the $\e$-thick part $\thick$.
\end{definition}

\subsection{Detection of collars}

The following lemma defines a quantity $l_\e$ that will help us detect $\e$-collars in the algorithm.

\begin{lemma}
\label{lem:ploop}
	Let $T$ be the Delaunay triangulation of a set of points $P$ on a hyperbolic surface $\surf$, $\e>0$ and $l_\e$ be the positive number such that $\cosh(\e)=\cosh^2(l_\e/2)$. If $P$ is an $\e$-packing and if there exists a geodesic loop on $\surf$ based at a point $p\in P$ that is shorter than $l_\e$, then $T$ contains a loop edge based at $p$ that is shorter than $l_\e$.
\end{lemma}

\begin{proof}
	Let $\g$ be the shortest geodesic loop based at $p$. In $\HH$, we consider a lift $\h \g$ between two lifts $\h p_1$ and $\h p_2$ of $p$. We will show that the circle $\h C$ of diameter the geodesic segment $\h p_1\h p_2$ is empty.

\begin{figure}[!ht]
	\centering
	\includegraphics{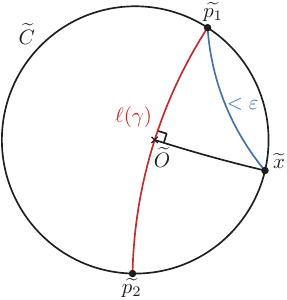}
	\caption{Illustration of the proof of Lemma~\ref{lem:ploop}.}
	\label{fig:ploop}
\end{figure}

	The points in the disk bounded by $\h C$ that are further away from $\h p_1$ and $\h p_2$ are those that lie both on the bisector of $\h p_1$ and $\h p_2$ and on $\h C$. Let $\h x$ be one of these points. The triangle formed by $\h x$, $\h p_1$ and the center $\h O$ of $\h C$ has a right angle at $\h O$. By Hyperbolic Pythagoras' Theorem,
	\[
	\cosh(\disth(\h x, \h p_1))
	= \cosh^2\left(\frac{\ell(\gamma)}{2}\right)
	< \cosh^2\left(\frac{l_\e}{2}\right)
	= \cosh \e,
	\]
	so $\disth(\h x,\h p_1)<\e$. Since $\g$ is the shortest curve based at $p$ and $P$ is an $\e$-packing, $\h C$ cannot contain another lift of $p$ nor a lift of another vertex of the Delaunay triangulation. Therefore, a Voronoi edge passes through the midpoint of $\h p_1$ and $\h p_2$ in the dual Voronoi diagram of the lifted triangulation, so $\g$ is an edge of $T$.
\end{proof}

Let us rewrite the equation $\cosh^2(l_\e/2)=\cosh \e$. Since $\cosh^2(l_\e/2)=\sinh^2(l_\e/2)+1$ and $\cosh \e=2\sinh^2(\e/2)+1$, we obtain $\sinh^2(l_\e/2)=2\sinh^2(\e/2)$, so $\sinh(l_\e/2)=\sqrt 2\sinh(\e/2)$. In particular, this equation allows us to easily see that $l_\e > \e$. Note that we implicitly used the fact that $\e$ and $l_\e$ are positive. We will reuse this fact in the following computations without further mention.

The following lemma will serve as a security tool to ensure the $\e$-packing property during the algorithm.

\begin{lemma}
\label{lem:lnsqrt}
	Let $\rho$ be a small curve. If $\e \leq \ln \sqrt 2$, then $\mathcal C(\rho, \e) \subset \mathcal C(\rho, l_\e) \subset \mathcal C(\rho)$ and the distance between the boundaries of $\mathcal C(\rho, l_\e)$ and $\mathcal C(\rho, \e)$ is greater than $\e$ (Figure~\ref{fig:lem_lnsqrt}).
\end{lemma}

\begin{figure}[!ht]
	\centering
	\includegraphics{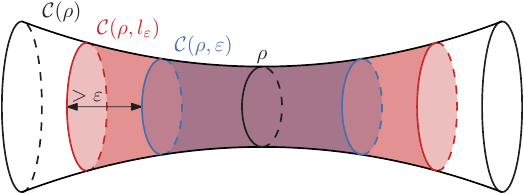}
	\caption{Illustration of Lemma~\ref{lem:lnsqrt}.}
	\label{fig:lem_lnsqrt}
\end{figure}

\begin{proof}
	The collar inclusion follows directly from Theorem~\ref{thm:collar} combined with the observation that $l_\e/2 < \arcsinh 1$ whenever $\e \leq \ln \sqrt 2$, completes the argument.

	Let $d_1$ be the distance between $\rho$ and the boundary of $\mathcal C(\rho, l_\e)$, and $d_2$ be the same distance for $\mathcal C(\rho, \e)$. By Theorem~\ref{thm:collar}, we have $\sinh(l_\e/2)=\sinh(\ell(\rho)/2)\cosh(d_1)$ and $\sinh(\e/2)=\sinh(\ell(\rho)/2)\cosh(d_2)$. It follows that:
	\[
		\left\{
		\begin{array}{l}
		d_1 = \arccosh\left(\dfrac{\sinh(l_\e/2)}{\sinh(\ell(\rho)/2)}\right)=\arccosh\left(\dfrac{\sqrt{2}\sinh(\e/2)}{\sinh(\ell(\rho)/2)}\right), \\
		d_2 = \arccosh\left(\dfrac{\sinh(\e/2)}{\sinh(\ell(\rho)/2)}\right).
		\end{array}
		\right.
	\]
	Since $\arccosh x = \ln\left(x+\sqrt{x^2-1}\right)=\ln x +\ln\left(1+\sqrt{1-1/x^2}\right)$ for all $x>1$, we obtain
	\[
		\left\{
		\begin{array}{l}
		d_1 = \ln \sqrt 2 + \ln\left(\dfrac{\sinh(\e/2)}{\sinh(\ell(\rho)/2)}\right)+\ln\left(1+\sqrt{1-\dfrac{\sinh^2(\ell(\rho)/2)}{2\sinh^2(\e/2)}}\right), \\
		d_2 = \ln\left(\dfrac{\sinh(\e/2)}{\sinh(\ell(\rho)/2)}\right)+\ln\left(1+\sqrt{1-\dfrac{\sinh^2(\ell(\rho)/2)}{\sinh^2(\e/2)}}\right).
	\end{array}
		\right.
	\]
	Combining both equations yields
	\[
	d_1-d_2 = \ln \sqrt 2+\ln\left(1+\sqrt{1-\frac{\sinh^2(\ell(\rho)/2)}{2\sinh^2(\e/2)}}\right)-\ln\left(1+\sqrt{1-\frac{\sinh^2(\ell(\rho)/2)}{\sinh^2(\e/2)}}\right)
	\]
	Therefore, $d_1-d_2>\ln \sqrt 2 \geq \e$.
\end{proof}

\subsection{Description of the algorithm}

Thanks to Lemmas~\ref{lem:ploop} and~\ref{lem:lnsqrt}, we have a tool to detect collars assuming that $\e \leq \ln \sqrt 2$. To compute a pseudo $\e$-net, we modify the standard $\e$-net algorithm: the core idea is to halt the algorithm when a point is added on an $l_\e$-collar and then "remove" the corresponding $\e$-collar instead of inserting this point. We use a tag to handle what happens in the $l_\e$-collars: the triangles passing through an $\e$-collar are tagged to both indicate that the collar has already been processed and to prevent the algorithm from destroying these triangles.

The input is the same as the $\e$-net algorithm: it is a Delaunay triangulation of the single vertex $b$. The case where $b$ is in a $l_\e$-collar requires a preprocessing step that is detailed in the remark after the algorithm.

\begin{figure}[!hb]
	\centering
	\includegraphics{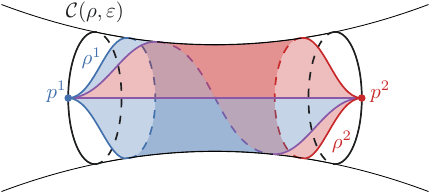}
	\caption{Step~\ref{step3} of the algorithm. The two colored triangles are the tagged triangles.}
	\label{fig:step3}
\end{figure}

\begin{algorithm}
\textcolor{white}{.}

\textbf{Input}: Delaunay triangulation of a single vertex $b$ and $0 < \e \leq \ln \sqrt 2$.

\textbf{Output}: Delaunay triangulation of a pseudo $\e$-net of $\surf$.

\begin{enumerate}
	\item\label{step1} Follow the $\e$-net algorithm until a point $p$ with an attached loop edge shorter than $l_\e$ is found. This is verified at each flip in constant time. The flip algorithm continues even if such a loop edge is found. This loop edge is homotopic to a closed geodesic $\rho$ and $p\in\mathcal C(\rho, l_\e)$.
	\item\label{step2} If $\rho$ is longer than $\e$, go back to Step~\ref{step1}. Note that this case could appear only if $\rho$ has length between $\e$ and $l_\e$. Else, compute the projection $p^\rho$ of $p$ on $\rho$ and locate it in the triangulation.
	\begin{enumerate}
		\item If $p^\rho$ lies in a tagged triangle, then $\mathcal C(\rho, \e)$ has already been processed\footnote{Of course, we could design an optimized method to verify whether a collar has already been processed. Since this method only adds a point location step, we accept this complexity for simplicity.}. Go back to Step~\ref{step1}.
		\item Else, cancel the step where $p$ was inserted. Insert consecutively two points $p^1$ and $p^2$ on each boundary component of $\mathcal C(\rho, \e)$. Go to Step~\ref{step3}.
	\end{enumerate}
	\item\label{step3} After the insertion of $p^1$ and $p^2$, the Delaunay triangulation contains two loop edges $\rho^1$ and $\rho^2$, respectively based at $p^1$ and $p^2$, which are both homotopic to $\rho$ (Figure~\ref{fig:step3}). Tag the two triangles in $\mathcal C(\rho, \e)$ incident to these loop edges. These triangles are then ignored by the algorithm when searching for a large triangle. Go back to Step~\ref{step1}.
\end{enumerate}
\end{algorithm}

\begin{remark}
  The initial point $b$ is in a $l_\e$-collar if and only if one of the edges of the initial Delaunay triangulation is shorter than $l_\e$. Without loss of generality, we can suppose that, in this case, it lies on the corresponding small curve. Indeed, if the base point $b$ of the Dirichlet domain $\D b$ is in the $l_\e$-collar of a small curve $\rho$, we can apply the algorithm proposed in~\cite{despreComputingDirichletDomain2023} to obtain a new Dirichlet domain $\D{b'}$ with $b'\in\rho$. To deal with the case where $b$ lies on a closed geodesic shorter than $\e$, we first insert the points $b^1$ and $b^2$ corresponding to the $\e$-collar of $\rho$ such that the geodesic line $b^1 b^2$ is orthogonal to $\rho$ at $b$. Then, we remove $b$ from the triangulation while maintaining the Delaunay property. The loop edges based at $b^1$ and $b^2$ are already in the triangulation. It suffices to remove the edges incident to $b$ and add the two edges joining $b^1$ and $b^2$ (like the purple ones in Figure~\ref{fig:step3}). Tag the two corresponding triangles as in Step~\ref{step3}.
\end{remark}

We need to verify that this algorithm terminates, outputs a pseudo $\e$-net of $\surf$ and that all the required computations can be performed in our model. Computations will be handled in the Poincaré disk $\HH$, in which a lift of a collar $\mathcal C(\rho)$ is a region $\{\h x\in\HH : \disth(\h x, \h \rho) \leq w(\rho)\}$, where $\h \rho$ is a lift of $\rho$. Since $\rho$ is a closed geodesic curve on $\surf$, $\h \rho$ is a geodesic segment of $\HH$. The two boundary components of the collar lift to two (non-geodesic) circular arcs that are equidistant to $\h \rho$. Their supporting circles both intersects the supporting geodesic line of $\h \rho$ on the boundary of the Poincaré disk. A lift of $\mathcal C(\rho)$ is therefore a piece of what Thurston called a \emph{banana} for obvious visual reasons. The same applies to lifts of $\e$-collars.

\begin{figure}[!ht]
	\centering
	\includegraphics[page=1]{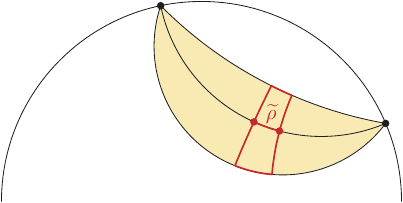}
	\caption{A banana (yellow) and a lift of a collar of a small curve $\rho$.}
	\label{fig:banane}
\end{figure}

\subsection{Correctness of the algorithm}

At each step, the set of vertices of the Delaunay triangulation is an $\e$-packing of $\thick$. Similarly to the standard algorithm, when a circumcenter of a large triangle is inserted on $\thick$, it is $\e$ far from all the vertices of the current Delaunay triangulation. Additionally, when a couple $\{p^1, p^2\}$ is inserted, both points are at distance at least $\e$ from the other vertices of the triangulation by Lemma~\ref{lem:lnsqrt}. Note that it is possible that a couple $\{p^1, p^2\}$ is separated by a distance smaller than $\e$ on $\surf$ if the corresponding small curve has length close to $\e$. However, on $\thick$, their distance is at least $\e$ because a path joining them would first have to go out of the corresponding $l_\e$-collar, whose boundary is at distance greater than $\e$ from either points by Lemma~\ref{lem:lnsqrt}. Moreover, in this case, Lemma~\ref{lem:ploop} still applies because every vertex that is not $p^1$ or $p^2$ is at distance $\e$ from the other vertices.

The following lemma shows that each large triangle that is not tagged will be processed. Together with the same arguments as for the standard $\e$-net algorithm, this proves that the procedure terminates and that the output is a pseudo $\e$-net of $\surf$.

\begin{lemma}
\label{lem:center}
	Let $\e \leq \ln \sqrt 2$ and $\surf$ be a hyperbolic surface. Let $T$ be a Delaunay triangulation of $\surf$ obtained at any step of the pseudo $\e$-net algorithm. Then the circumcenter of any triangle that is not tagged is either in $\thick$ or in an undiscovered thin part.
\end{lemma}

\begin{proof}
	Let $C$ be the circumcircle of a large triangle $t$ of $T$ that is not tagged such that its circumcenter is not in an undiscovered thin part. Suppose also that its center is not in $\thick$. In other words, it is in an already handled $\e$-collar. The circle $C$ intersects the boundary of this $\e$-collar. Let $p$ be the vertex with an attached loop edge of length $\e$ on this boundary component. Consider a lift $\h C$ of $C$. Since it is an empty circle, it passes between the two lifts $\h p_1$ and $\h p_2$ of $p$ such that the geodesic segment $\h p_1\h p_2$ is a lift of the loop edge.

	\begin{figure}[!ht]
		\begin{subfigure}
			{0.5\textwidth}
			\centering
			\includegraphics[page=2]{covering_proof}
		\end{subfigure}%
		\begin{subfigure}
			{0.5\textwidth}
			\centering
			\includegraphics[page=1]{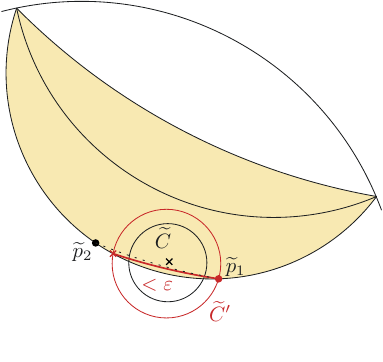}
		\end{subfigure}
		\caption{The two cases of the proof of Lemma~\ref{lem:center}. Left: $\h C$ is not drawn, but its center is represented as the black cross.}
		\label{fig:covering_proof}
	\end{figure}

	If the center $\h c$ of $\h C$ lies between the lift of the boundary component of $p$ and the geodesic segment $\h p_1 \h p_2$, then it lies in the disk $\h D$ of diameter $\h p_1 \h p_2$ (Figure~\ref{fig:covering_proof}, left). The diameter of $\h D$ is $\e$, and since $\h C$ has radius larger than $\e$, it must contain $\h p_1$ and $\h p_2$ in its interior. This contradicts the fact that $C$ is empty.

	Else, let $\h C'$ be the circle obtained by increasing the radius of $\h C$ until it touches $\h p_1$ or $\h p_2$ (Figure~\ref{fig:covering_proof}, right). Up to renumbering, suppose that it is $\h p_1$. The intersection between $\h C'$ and the lift of the boundary component of the $\e$-collar yields a chord shorter than $\e$. The triangle $t$ has at least one vertex that is not $p$. This vertex has a lift on $\h C$ that is in the smaller piece of the interior of $\h C'$ delimited by the chord. It is thus at distance less than $\e$ from $\h p_1$, which contradicts the fact that the vertices of $T$ form an $\e$-packing of $\thick$.
\end{proof}

\subsection{Computation details \& complexity analysis}

\begin{lemma}
\label{lem:bananasplit}
	Let $T$ be the Delaunay triangulation of a set of points $P$ on a hyperbolic surface $\surf$ and $\e > 0$. If\hspace{2pt} $T$ contains a loop edge shorter than $l_\e$ based at a vertex $p$, then the lift of an $\e$-collar containing $\h p_o$ can be computed in constant time.
\end{lemma}

\begin{proof}
	All the notations of this Lemma are illustrated in Figure~\ref{fig:bananasplit}.
	\begin{figure}[!ht]
		\centering
		\includegraphics[page=2]{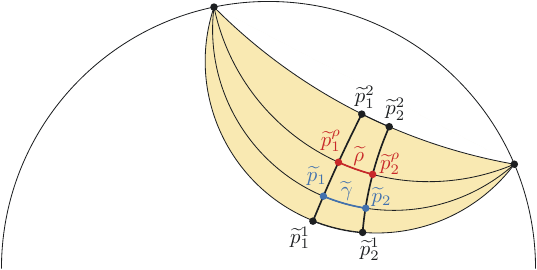}
		\caption{Illustration of the notations of the proof of Lemma~\ref{lem:bananasplit}.}
		\label{fig:bananasplit}
	\end{figure}

	Let $\g$ be the loop edge around $p$, $\h \g$ be its lift with vertex $\h p_1=\h p_o$ and $\h p_2$ be the second vertex of $\h \g$. The translation $\tau$ that sends $\h p_1$ to $\h p_2$ is stored in the data structure. Its fixed points can then be computed in constant time as the two solutions of the quadratic equation $\tau(z) = z, z\in\mathbb{C}$. From this, its axis is computed as the geodesic line whose points at infinity are the fixed points of $\tau$. The length of the associated curve $\rho$ on $\surf$ is computed as well: $\ell(\rho)=\ell(\tau)=2\arccosh(|\tr(\tau)|/2)$.

	The two geodesic lines in the Poincaré disk $\HH$ that are orthogonal to $\rho$ and go through $\h p_1$ and $\h p_2$ are computed. On those lines, we want to compute the points $\h p_1^1$, $\h p_1^2$, $\h p_2^1$ and $\h p_2^2$ that are on the boundary of the banana corresponding to $\mathcal C(\rho, \e)$. By Theorem~\ref{thm:collar}, the distances from these points and the points $\h p_1^\rho$, $\h p_2^\rho$ that lie on $\h \rho$ and on the two orthogonal lines is $\arccosh(\sinh(\e/2)/\sinh(\ell(\rho)/2))$.

	All the computations are performed in $\HH$ in constant time in the real RAM model.
\end{proof}

We keep the notations from the proof of the previous lemma and show the following:
\begin{lemma}
\label{lem:locatecollar}
	Let $i$ be the number of vertices in the Delaunay triangulation at step~\ref{step2} of the algorithm. If $\e < 2\arcsinh 1$, the points $p^1$ and $p^2$ that lift to $\h p_1^1$ and $\h p_1^2$ (respectively) can be located in $O(i)$ time in the Delaunay triangulation.
\end{lemma}

\begin{proof}
\begin{figure}[!ht]
	\centering
	\includegraphics[page=1]{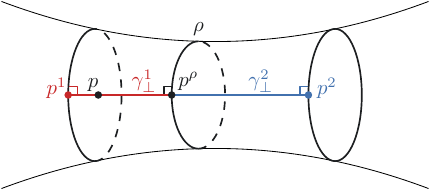}
	\caption{Illustration of the notations of the proof of Lemma~\ref{lem:locatecollar}.}
	\label{fig:gamma_perp}
\end{figure}
	Let $p^1$ and $p^2$ be the points of $\surf$ on which $\h p_1^1$ and $\h p_1^2$ respectively project. Let $p^\rho$ be the projection of $p$ on $\rho$, which projects on $p_1^\rho$. Like for the complexity of the location of a circumcenter, we want prove that the orthogonal geodesic segment $\h \gamma_\perp^1$ from $\h p_1^1$ to $\h p_1^\rho$ projects on a shortest path $\gamma_\perp^1$ on $\surf$ (Figure~\ref{fig:gamma_perp}).

	By Theorem~\ref{thm:collar}, and since $\e < 2\arcsinh 1$, we know that the $\e$-collars are embedded on $\surf$. Then, $\gamma_\perp^1$ is a shortest path inside the collar and no shortest path goes outside, else it would have to pass through the other end of the collar and therefore be longer than $\gamma_\perp^1$. The same holds for the symmetric path $\gamma_\perp^2$. Then, $\h p_1^1$ and $\h p_1^2$ can be located in the tiling $\{\g\h D_o\}_{\g\in\Gamma}$ from $\h p_1$ in $O(g^2)$ time by walking along the geodesic segments $\gamma_\perp^1$ and $\gamma_\perp^2$. Note that the copy in which $\h p_1$ lies is already known since it is the lift of the circumcenter $p$ that has been inserted in the triangulation. Their lifts in $\h D_o$ are then computed in constant time. Finally, the triangles in which they lie are found in $O(i)$ time.
\end{proof}

\begin{theorem}
\label{thm:complexite_bananes}
	Let $\surf$ be a hyperbolic surface given by a Delaunay triangulation on a single vertex (or, equivalently, by a Dirichlet domain) and let $\e\leq\ln\sqrt 2$ be a positive number. A pseudo $\e$-net of $\surf$ can be computed in $O(1/\e^4)$ time, where the hidden constant depends on $\surf$.
\end{theorem}

\begin{proof}
	We keep the same notations as in the previous proofs.

	We first prove an upper bound on the number of inserted points. If the algorithm terminates without reaching step two, $\surf$ is $\e$-thick by Lemma~\ref{lem:ploop}. The total number of inserted points is then smaller than $16(g-1)/\e^2$ by Proposition~\ref{prop:nb_points}. Otherwise, the situation is similar except that there are no vertices in the $\e$-collars. For each $\e$-collar, a point have been inserted and then removed. Two other points have instead been inserted on the two boundary components. Since there are at most $3g-3$ $\e$-collars, the number of inserted points is then upper bounded by $16(g-1)/\e^2 + 3g - 3$.

	Lemmas~\ref{lem:bananasplit} and~\ref{lem:locatecollar} imply that inserting a couple $\{p^1, p^2\}$ for a given $\e$-collar is as costly as inserting a point like in the $\e$-net algorithm (only the constant changes). Similarly, locating $p^\rho$ at step~\ref{step2} of the algorithm only increases the constant. We can repeat similar computations as in the proof of Lemma~\ref{lem:nbflips} to obtain a total complexity of the algorithm of $O(1/\e^4)$ time.
\end{proof}

\begin{remark} The algorithm can be used with $\e=2\arcsinh 1$ to compute a thick-thin decomposition of the surface. The proof of Theorem~\ref{thm:complexite_bananes} applies to this situation but the result might not be an $\e$-packing around the boundaries of the $\e$-collars because Lemma~\ref{lem:lnsqrt} does not apply in this case.
\end{remark}
\begin{figure}
    \centering
    \includegraphics[width=0.5\linewidth]{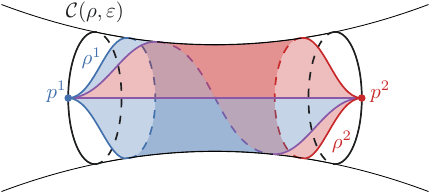}
    \caption{Enter Caption}
    \label{fig:placeholder}
\end{figure}

\section{Computing the systole of $\surf$ and its length spectrum}

Short geodesics on $\surf$ can be identified directly from a Delaunay triangulation by exploiting the fact that a short curve cannot cross Delaunay edges too many times, following an approach closely related to that of Akrout~\cite{akrout2006}. We describe this method in Appendix~\ref{sec:Akrout}.

In this work, we propose a different approach based on the following observation. Every curve $\gamma$ of length less than $L$ is passes close to a point in the $\e$-net. This point determines a based geodesic loop homotopic to $\gamma$, and whose length is comparable to that of $\gamma$ and is therefore expected to remain below $L$ (up to a small error). Consequently, computing the length spectrum of $\surf$ up to length $L$ reduces to finding pairs of lifts of a common point on $\surf$ whose mutual distance is at most $L$, up to a controlled error.

The existence of an $\e$-net alone is not sufficient for this purpose, since there is no immediate way to identify all such pairs of lifts. The crucial additional ingredient is that our $\e$-nets are equipped, by construction, with Delaunay triangulations. These triangulations satisfy the key property that each edge has length at most $2\e$, providing enough combinatorial structure to locate the relevant pairs efficiently.

Another difficulty stems from the fact that the multiplicities in the length spectrum of a hyperbolic surface are unbounded (see \cite{randol1980} and Remark \ref{rem:multiplicity} below). Consequently, distinct free homotopy classes may correspond to closed geodesics of the same length. To handle this phenomenon, we rely on the optimal homotopy test of Lazarus and Rivaud~\cite{lazarus2012}.

The main result underlying this approach is the following theorem.

\begin{theorem}\label{thm:spectrum}
Let $\e>0$, let $\surf$ be an $\e$-thick hyperbolic surface, and let $T$ be a Delaunay triangulation over an $\e$-net of $\surf$. The length spectrum of $\surf$ up to length $L$ can be computed in $O(g\,m\,L\,e^L)$ time, where $m$ denotes the maximum multiplicity of a length in the spectrum below $L$.
\end{theorem}

\begin{remark}\label{rem:multiplicity}
    The maximum multiplicity in the above statement could be replaced by a bound which depends on $g$ and $L$. The most obvious bound would be to bound the multiplicity by the total number of geodesics up to length $L$. Using \cite[Lemma 6.6.4]{buserGeometrySpectraCompact2010}) provides a bound of $m\leq (g+1)e^{L+6}$. It is worth noting that, while this bound might seem particularly naive, multiplicities are always unbounded in a length spectrum and there are surfaces with multiplicities that grow on the order of $e^{\delta L}$ for $\delta > \frac{1}{2}$ \cite{haoLengthSetsClosed2026}.
\end{remark}

\subsection{Structural properties}
Before proving the theorem above, we establish several preliminary results.
\begin{lemma}\label{lem:eball}
      Let $B$ be a ball of $\HH$ of radius $R$ and $\e$ be a positive number. Then, $B$ contains at most $\frac{\sinh^2(R/2+\e/4)}{\sinh^2(\e/4)}$ points of a $\e$-net of $\HH$.
\end{lemma}
\begin{proof}
    We proceed similarly to Section~\ref{sec:bound}. Let $p$ be a point of the $\e$-net contained in $B$, then, $B(p,\e/2)$ is empty and its area is $4\pi\cdot\sinh^2(\e/4)$. $B(p,\e/2)$ is also included in a ball $B_\e$ that has the same center as $B$ and whose radius is $R+\e/2$ (see Figure~\ref{fig:eball}). The sum of the area of the $B(p,\e/2)$ for all $p\in B$ is then smaller than the area of $B_\e$. If $N$ is the number of points of the $\e$-net in $B$, then we obtain: $N\cdot4\pi\cdot\sinh^2(\e/4)\le4\pi\cdot\sinh^2(R/2+\e/4)$. The result follows.
\end{proof}

\begin{figure}[!ht]
    \centering
    \includegraphics{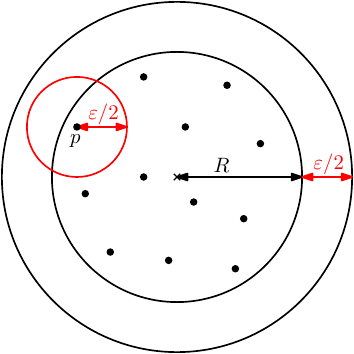}
    \caption{Illustration of the proof of Lemma~\ref{lem:eball}.}
    \label{fig:eball}
\end{figure}

\begin{lemma}\label{lem:BFS}
    Let $\e$ and $R$ be 2 positive numbers such that $R>3\e$, $T$ a Delaunay triangulation over an $\e$-net of $\HH$, and $p$ a vertex of $T$. We denote by $T_R$ the triangulation that is the restriction of $T$ to $B(p,R)$. Then, the connected component of $T_R$ that contains $p$ also contains all the points of $T$ in $B(p,R-3\e)$.
\end{lemma}
\begin{proof}
Let $x$ be a point of $T_R$ that is not in the connected component of $p$ and let $C_x$ be the connected component of $x$. All the vertices in the boundary of $C_x$ admit at least an edge in $T$ that goes outside of $B(p,R)$. Then, since the edges of $T$ are smaller than $2\e$, all the vertices are distance smaller than $2\e$ from the boundary of $B(p,R)$. The closest point $x_p$ from $p$ in $C_x$ (considered globally) is either a vertex or a point on an edge that links two points on the boundary of $C_x$. If $x_p$ is a vertex of then $d(x_p, p)\leq R-2\e$ by the above. Else, it is a point on the interior of an edge $(u, v)$. Suppose that it is closer to $u$ than to $v$. Since $(u, v)$ has length at most $2\e$, $d(u, x_p)\geq \e$. The triangle inequality then yields $d(p, x_p)\geq d(p, u)-d(u, x_p) \geq (R-2\e) - \e = R-3\e$. See Figure~\ref{fig:BFS} for an illustration.
\end{proof}

\begin{figure}[!ht]
    \centering
    \includegraphics{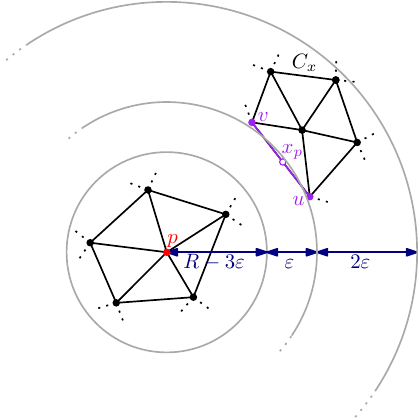}
    \caption{Illustration of the case where $x_p$ is a point on the interior of an edge.}
    \label{fig:BFS}
\end{figure}

\begin{lemma}\label{lem:approxlength}
    Let $\surf$ be a surface, $\e$ be a positive number and $T$ a Delaunay triangulation over an $\e$-net of $\surf$. If $\surf$ contains a curve $\gamma$ of length $L$ then, there is a point $p$ of $T$ that admits a curve homotopic to $\gamma$ smaller than $L_p$ verifying $\sinh\left(\frac{1}{2}L_p\right)=\sinh\left(\frac{1}{2}L\right)\cdot\cosh(\e)$.
\end{lemma}
\begin{proof}
    Since $T$ is defined on an $\e$-net, it has a point at distance $d\le\e$. Then, Theorem~\ref{thm:collar} (iii) applies to this situation and gives the result.
\end{proof}

\subsection{The algorithm and proof of its complexity}
We can now compute the length spectrum of $\surf$ up to a given length $L$ using the following algorithm.

\begin{algorithm}
\textcolor{white}{.}

\textbf{Input}: A Delaunay triangulation over an $\e$-net of $\surf$, and a real number $L>0$.

\textbf{Output}: The length spectrum of $\surf$ up to length $L$.

\begin{enumerate}
\item\label{alg:preprocess} Perform the preprocessing of Lazarus and Rivaud~\cite{lazarus2012} to obtain a combinatorial representation $M$ of $\surf$.
\item\label{alg:init} Initialize an empty list $\mathcal{L}$ storing pairs consisting of the length of a closed geodesic and a canonical representative of its free homotopy class in $M$.
\item\label{alg:loop} For each vertex $p$ of the Delaunay triangulation $T$:
\begin{enumerate}
    \item Choose a lift $\tilde{p}$ of $p$ and compute the intersection of the lifted triangulation $\tilde{T}$ with the ball $B(\tilde{p},L_p+3\e)$ using a breadth-first search.
    
    \item Whenever the search reaches another lift of $p$, compute the length of the corresponding closed geodesic (as in~\cite{despreComputingDirichletDomain2023}) together with the canonical representative of its free homotopy class in $M$. If this representative is not already present in $\mathcal{L}$, add the corresponding pair.
\end{enumerate}
\item Return $\mathcal{L}$.
\end{enumerate}
\end{algorithm}

\begin{proof}[Proof of Theorem~\ref{thm:spectrum}]
Let $\gamma$ be a closed geodesic on $\surf$ of length less than $L$. By Lemma~\ref{lem:approxlength}, there exists a vertex $p$ of $T$ such that the pointed geodesic $\gamma_p$ based at $p$ and freely homotopic to $\gamma$ has length at most $L_p$. Choose a lift $\tilde p$ of $p$ in $\HH$. By Lemma~\ref{lem:BFS}, a breadth-first search in $\tilde T$ restricted to the ball $B(\tilde p,L_p+3\e)$ reaches another lift $\tilde p'$ of $p$ such that the geodesic segment joining $\tilde p$ and $\tilde p'$ projects to $\gamma_p$.

By Lemma~\ref{lem:eball}, this search visits at most
$\frac{\sinh^2(L_p/2+7\e/4)}{\sinh^2(\e/4)}
=O(\sinh^2(L_p/2))
=O(\sinh^2(L/2))$
vertices. To recover the entire length spectrum up to length $L$, the procedure is repeated for every vertex of $T$. Since $T$ is $\e$-thick, it contains at most $16(g-1)/\e^2=O(g)$ vertices.

For each pair of lifts of a vertex, the algorithm performs three operations: it computes the length of the corresponding closed geodesic, computes a canonical representative of its free homotopy class in $M$, and compares this representative with those of previously discovered geodesics having the same length. The first operation takes constant time, the second requires $O(L)$ time~\cite{lazarus2012}, and the total cost of the homotopy tests is $O(mL)$, where $m$ denotes the maximum multiplicity of a length in the spectrum below $L$. Therefore, the overall running time is $O(gmLe^{L})$.
\end{proof}

\subsection{Pseudo $\e$-net version}
We have described a strategy for computing the length spectrum of $\surf$ by first constructing a $1$-net, which either yields its systole or certifies that $\sys>1$. If $\sys<1$, we then construct a $\sys$-net and apply Theorem~\ref{thm:spectrum}. However, when $\sys$ is small, it is more efficient to construct a $\log(\sqrt{2})$-net directly and apply the following variant of Theorem~\ref{thm:spectrum} to the resulting Delaunay triangulation. The choice of the threshold $1$ proposed for the thick case is arbitrary; determining the optimal value is an interesting practical problem.

\begin{theorem}[Variant of Theorem~\ref{thm:spectrum}]\label{thm:spectrum2}
    Let $\surf$ be surface and $T$ a Delaunay triangulation over a pseudo $\log(\sqrt{2})$-net of $\surf$. We can compute the length of all the closed free geodesics of $\surf$ smaller than $L$ in $O(g\,m\,L\,e^L)$ time, where $m$ denotes the maximum multiplicity of a length in the spectrum below $L$.
\end{theorem}
\begin{proof}
    We just need to adapt the proof of Lemma~\ref{lem:BFS} to this situation since everything else remains unchanged. So, let $p$ be a point of the $\log(\sqrt{2})$-net and $x$ a point of $T$ restricted to $B(p,R)$ for some $R>0$. We denote by $C_x$ the connected components of $T$ containing $x$ and assume that $C_x$ does not contain $p$. The edges of $T$ are all smaller than $2\log(\sqrt{2})$ apart from the edges that links to points on the boundary of a $\log(\sqrt{2})$-thin cylinder. We call such points cylinder points. We consider the closest point $x_p$ from $C_x$ to $p$. If $x_p$ is a non-cylinder vertex or a point of an edge between two non-cylinder vertex the $x_p$ is outside $B(p,R-3\log(\sqrt{2}))$ using the same reasoning as the proof of the initial lemma. Now, let us assume that $x_p$ is a cylinder vertex and assume that $x_p$ is inside $B(p,R-2\log(\sqrt{2}))$. Then, all non-cylinder vertices of the star around $x_p$ are in $C_x$ and, thus, the path from $p$ to $x_p$ hit $x_p$ from the interior of the thin-cylinder. Since $p$ is outside this cylinder it is closer to the other extremity vertex of the cylinder, a contradiction because this vertex is in $C_x$. If $x_p$ is on an edge from a cylinder point to an non-cylinder point, it is on an-edge smaller than $2\log(\sqrt{2})$ with an extremity outside $B(p,R-2\log(\sqrt{2}))$ then it is outside $B(p,R-4\log(\sqrt{2}))$. The last case remaining is when $x_p$ lie on an edge between two cylinder points. Those two points are necessarily on the same side of the cylinder (they are two lifts of the same point of $\surf$ unless the closest point from $p$ would be an extremity of the edge. Assume that one of the cylinder point $x_c$ of the edge $T$ is in $B(p,R-2\log(\sqrt{2}))$, then the path from $p$ to $x_p$ hit $x_p$ from the interior of the thin-cylinder. If the cylinder is wider than $2\log(\sqrt{2})$, then $p$ is closer to the other extremity vertex of the cylinder, a contradiction. Then, $x_c$ is connected to some point outside $B(p,R)$ by an edge that is smaller than $3\log(\sqrt{2})$ and thus it is outside $B(p,R-3\log(\sqrt{2}))$. It implies that $x_p$ is outside $B(p,R-4\log(\sqrt{2}))$. SO, in all cases $x_p$ and thus, $x$ is outside $B(p,R-4\log(\sqrt{2}))$. So, Lemma~\ref{lem:BFS} apply to this situation with $\e=\log(\sqrt{2})$ and a marge of $4\e$ instead of $3\e$.
\end{proof}

Alternatively, starting from a pseudo $\log(\sqrt{2})$-net, we can compute the systole of $\surf$ in time polynomial in the genus $g$. This does not yield a polynomial-time algorithm for computing the systole, since constructing the pseudo-net itself has complexity exponential in $g$. Nevertheless, we believe that the practical complexity of the flip algorithm used to construct the pseudo-net is significantly lower, making it worthwhile to analyze separately the complexity of the second stage, which computes $\sys$ from the pseudo-net.

The strategy is as follows. If the construction of the pseudo-net reveals a closed geodesic of length smaller than $\log(\sqrt{2})$, then the systole is necessarily one of these curves. Otherwise, $\surf$ is $\log(\sqrt{2})$-thick, and we can apply the algorithm of Theorem~\ref{thm:spectrum2}. Since we are only interested in the shortest geodesic, there is no need to handle multiplicities. This yields the following result, together with the fact that the systole of a closed surface of genus $g$ is at most $2\log(g)+8$ (see, for instance,~\cite{fanoni-parlier}).

\begin{corollary}
    Let $\surf$ be an $\log(\sqrt{2})$-thick surface and $T$ a Delaunay triangulation over a pseudo $\log(\sqrt{2})$-net of $\surf$. We can compute the length of the systole of $\surf$ in $O(g^2)$ time.
\end{corollary}

\newpage

\bibliography{biblio}

\appendix

\section{Lemma for the proof of correctness (Section~\ref{sec:firstalgo})}

\begin{lemma}\label{lem:point_triangle}
	Let $\h\Delta$ be a triangle of $\HH$ and $\h x\in \h\Delta$. Denote $r$ as the radius of the circumcircle of $\h\Delta$. Then there exists a vertex of $\h\Delta$ whose distance to $\h x$ is at most $r$.
\end{lemma}

This lemma holds in both the Euclidean and the hyperbolic planes, as the following proof uses arguments that function in both settings.

\begin{proof}
	Let $D_{\h x}$ be the closed disk of radius $r$ centered in $\h x$, $D_{\h\Delta}$ be the circumdisk of $\h\Delta$, and $\h c$ be its circumcenter. If $\h x=\h c$, the two disks are equal and the result is trivial.
				
	Suppose that $\h x\neq \h c$. Denote $C_{\h\Delta}$ (resp. $C_{\h x}$) as the circle bounding the disk $D_{\h \Delta}$ (resp. $D_{\h x}$). The distance between $\h x$ and $\h c$ is at most $r$ because $\h x$ lies in $\h\Delta$, so $\h x\in D_{\h\Delta}$ and $\h c\in D_{\h x}$. Therefore, the two disks intersect and their intersection is not a singleton. Moreover, $D_{\h\Delta} \neq D_{\h x}$, so $C_{\h \Delta}$ and $C_{\h x}$ intersect exactly twice (Figure~\ref{fig:lemme_triangle}). The three vertices of $\h\Delta$ lie on the circle $C_{\h\Delta}$. But if their distance to $\h x$ is greater than $r$, i.e. if they belong to $C_{\h\Delta} \setminus D_{\h x}$, then the geodesic passing through the two intersection points of the circles, that is, the perpendicular bisector between $\h x$ and $\h c$, separates $\h x$ from $\h\Delta$, so $\h x$ cannot be in $\h\Delta$ (Figure~\ref{fig:lemme_triangle}, left). Thus, at least one vertex of $\h\Delta$ must lie on $C_{\h\Delta} \cap D_{\h x}$, that is, at distance at most $r$ from $x$.
	\begin{figure}[!ht]
		\centering
		\begin{subfigure}{0.5\textwidth}
			\centering
			\includegraphics[page=1]{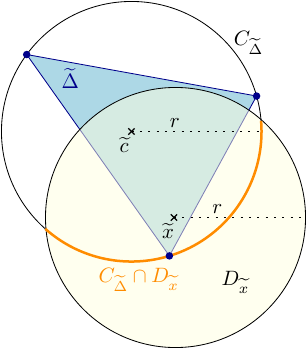}
		\end{subfigure}%
		\begin{subfigure}{0.5\textwidth}
			\centering
			\includegraphics[page=2]{lemme_point_triangle.pdf}
		\end{subfigure}
		\caption{Illustration of the proof of Lemma~\ref{lem:point_triangle} in Euclidean geometry (for simplicity). The points $\h c$ and $\h x$ belong in each other's disk of radius $r$. Right: If no vertex of $\h \Delta$ lie on $C_{\h\Delta}\cap D_{\h x}$ (bold), then $\h x$ is separated from $\h\Delta$ by the dashed geodesic.}\label{fig:lemme_triangle}
	\end{figure}
\end{proof}

\section{Computing the systole directly from a Delaunay triangulation}\label{sec:Akrout}
It is interesting to consider this strategy since it is very simple and independent of the geometry of $\surf$. However, the best complexity we obtain is not satisfying and it does not generalize to the study of the full spectrum of the surface. 
\begin{proposition}
Let $\surf$ be a surface given by a Delaunay triangulation on a single vertex. We can compute the systole of $\surf$, along with all of its small curves, in  $O\left(g\cdot2^{48g}\right)$ time.
\end{proposition}
\begin{proof}
Systoles are geodesically convex subsets, meaning that any shortest path between two of its points lies on the systole itself. In particular, the two arcs of the systole between two opposite points are shortest paths. As such, it can cross an edge of a Delaunay triangulation at most 4 times~\cite{despreRepresentingInfiniteHyperbolic2021}. Due to their collar width, the same argument applies to short curves.
In addition, a triangulation with a single vertex has $12g-6$ edges. We now use normal coordinates with respect to a Delaunay triangulation. As a systole can cross the triangulation at most 4 times, this gives a path in the triangulation of combinatorial length smaller than $48g-24$ corresponding to the systole.

We will use a somewhat similar strategy to Akrout~\cite{akrout2006} by performing an exhaustive search of homotopy classes seen as walks along sides of triangles (Akrout's algorithm uses pants decompositions instead of triangulations). More concretely, for each starting triangle, we compute all reduced walks of length smaller than $48g-24$. There are at most $2^{48g-24}$ such walks per triangle which gives a total of $O\left(g\cdot2^{48g-24}\right)$ admissible walks. The complexity of the algorithm follows, since computing the length of the closed geodesic in a free homotopy class is doable in constant time by a similar reasoning as the one in the proof of lemma~\ref{lem:bananasplit}.
\end{proof}

\end{document}